\documentclass[11pt]{article}   	                 
\usepackage{amsmath,amsthm, amssymb,amsfonts,amscd,mathabx}

\usepackage[colorlinks=true,linkcolor=blue,citecolor=red]{hyperref}
\usepackage[pdftex]{graphicx}                 
\usepackage{wrapfig}                               
\usepackage[font=footnotesize]{caption}                     
\usepackage{float}
\usepackage[symbol]{footmisc}
\counterwithin{figure}{section}
\usepackage{subcaption}

\let\OLDitemize\itemize
\renewcommand\itemize{\OLDitemize\addtolength{\itemsep}{-5pt}}

\numberwithin{equation}{section}
\newtheorem{theorem}{Theorem}[section]

\newtheorem{lemma}[theorem]{Lemma}
\newtheorem{proposition}[theorem]{Proposition}
\newtheorem{definition}[theorem]{Definition}

{\begin{trivlist}\item[]\textbf{Proof#1 }}%
{\qed\end{trivlist}}

 {\begin{trivlist}\item[]\textbf{Acknowledgments }}{\end{trivlist}}

\newcommand{\R}{\mathbb{R}}

\newcommand{\N}{\mathbb{N}}
\newcommand{\Z}{\mathbb{Z}}

\newcommand{\Rmnum}[1]{\uppercase\expandafter{\romannumeral #1\relax}}

\newcommand{\rmO}{\mathrm{O}}

\newcommand{\rmd}{\mathrm{d}}
\newcommand{\rme}{\mathrm{e}}

\renewcommand{\Re}{\mathrm{Re}\,}

\renewcommand{\leq}{\leqslant}
\renewcommand{\geq}{\geqslant}

\def\eps{\varepsilon}

\newfam\bifam
\font\tenbi=cmmib10 scaled \magstep1 \font\sevenbi=cmmib10 at 11pt
\font\fivebi=cmmib10 at 6pt \textfont\bifam = \tenbi
\scriptfont\bifam = \sevenbi \scriptscriptfont\bifam= \fivebi

\DeclareMathOperator{\floor}{floor}

\DeclareMathOperator\erf{erf}

\begin{document}

\begin{center}
{\fontsize{14}{14}\fontfamily{cmr}\fontseries{b}\selectfont{
Selection of Marginally Stable Patterns by Boundary Conditions
\footnote[4]{
The authors acknowledges partial support through grant NSF-DMS-2506837
}}}\\[0.2in]
Jacob Daniels$^{1}$, Simon Eskin$^{2}$, Tel Fahrendholz$^{1}$, and  Arnd Scheel$^{1}$\\[0.1in]
\textit{\footnotesize
$^1$School of Mathematics, University of Minnesota, Minneapolis, 206 Church St SE, MN 55455, USA\\
$^2$   Department of Mathematics, Applied Mathematics, and Statistics, Case Western Reserve University,
10900 Euclid Avenue, Cleveland, OH  44106, USA\\
}
\end{center}

\begin{abstract}
\noindent
    We study the effect of boundary conditions on the selection of wavenumbers in pattern-forming systems. Prior work established that, for a large class of boundary conditions, equilibria only within a narrow band of wavenumbers are compatible with the boundary. We demonstrate here that for a large class of initial conditions, solutions evolve towards an extremal member of this narrow band of equilibria. In fact, this extremal member is linearly stable but nonlinearly unstable, and the convergence is realized through a series of snapping events at the boundary. We observe this behavior in simulations of the Swift-Hohenberg equation when the pattern forming region is confined to a large spatial interval through a parameter quench. We prove the convergence in a phase-diffusion approximation using comparison principles and astutely chosen subsolutions that mimic the snapping behavior. 
\end{abstract}
\section{Introduction}
Patterns forming in domains with a large aspect ratio usually lead to complex dynamics with many stable and unstable equilibria. In the limit of an unbounded domain, for us here the real line $x\in\R$, patterns come typically in two-parameter families, parameterized by a phase and a wavenumber. In typical examples, equilibria are of the form $u(x)\sim u_\mathrm{per}(kx-\varphi;k)$, $u_\mathrm{per}(\xi+2\pi;k)=u_\mathrm{per}(\xi;k)$, with     wavenumber $k$ in an admissible band $k\in[k_-,k_+]$ and phase $\varphi\in\R/(2\pi\Z)$. While typically not all those patterns are stable, restricting to stable patterns commonly still allows for a band $[k_{-,\mathrm{s}},k_{+,\mathrm{s}}]\subsetneq [k_-,k_+] $ of wavenumbers. As in many other complex systems, one then would wish to predict which of these patterns is in fact \emph{selected}, when restricting to a particular class of initial conditions. 

\smallskip\textbf{Displacement-strain relations.} In real systems, these perfect periodic patterns are disturbed by boundaries and impurities in the medium. Our focus here is on the effect of boundaries. In large bounded systems, $|x|<L$, the number of equilibria typically grows linearly in $L$. A conceptual approach to describing  equilibria in large domains was proposed and carried out in several examples in~\cite{MorisseyScheel}, starting with the description of equilibria in semi-bounded domains $x\in(0,\infty)$ or $x\in(-\infty,0)$. It was shown there that equilibria can be understood as exponentially localized boundary layers that converge to a periodic pattern $u_\mathrm{per}(kx-\varphi;k)$ far away from the boundary, where $\varphi$ and $k$ satisfy a \emph{displacement-strain relation}. In many situations, this relation takes the simple form $k=d(\varphi)$ where $d$ is $2\pi$-periodic.

\smallskip\textbf{Swift-Hohenberg with parameter quench.} A somewhat universal example are systems near a Turing instability, when the driving parameter is small positive in a half space $\{x>0\}$ and has the opposite sign in the complement $\{x<0\}$. 

To be specific, consider the Swift-Hohenberg equation with a parameter quench, 
\begin{equation}\label{e:SH0}
    \partial_t u = - (\partial_{xx}+1)^2 u + \eps^2 \mathrm{sign}(x) u - u^3,\qquad x\in \R,
\end{equation}
so that patterns form in $\{x>0\}$ and are suppressed in $\{x<0\}$. Results from~\cite{weinburd} show that for small $\eps$, the displacement-strain relations is of the form $k=1+\frac{\eps^2}{16}\cos(2\varphi)+\rmO(\eps^3)$, that is, there exists a family of solutions $u_*$ with $u_*(x)\to 0$ as $x\to -\infty$ and $|u_*(x)-u_\mathrm{per}(kx-\varphi;k)|\to 0 $ as $x\to\infty$ when $k=d(\varphi)$. Here, $u_\mathrm{per}(kx;k)\sim \frac{2\eps}{\sqrt{3(1-4\kappa^2})}\cos((1+\kappa\eps)x)$, $k=1+\kappa\eps$, $|\kappa|<\frac{1}{2}$.  In particular, the parameter step restricts the band of wavenumbers at leading order to $k\in[1-\frac{\eps^2}{16},1+\frac{\eps^2}{16}]$. This interval is significantly smaller than the interval of stable wavenumbers, at leading order $k\in[k_{-,\mathrm{s}},k_{+,\mathrm{s}}]=[1-\frac{\eps}{2\sqrt{3}},1+\frac{\eps}{2\sqrt{3}}]$. Results in~\cite{MorisseyScheel} show that patterns undergo an exchange of stability at the extrema. More specifically, there is a sign $\iota\in\{\pm1\}$ that can be computed in specific examples such that equilibria where $\iota d'(\varphi)>0$ are stable and equilibria with $\iota d'(\varphi) <0$ are unstable. The pattern with maximal or minimal $k$ is, in this respect, marginally stable. 

The same type of displacement-strain relation also occurs when the Swift-Hohenberg equation is considered on $x>0$ with boundary conditions at $x=0$, 
\[
    \partial_t u = - (\partial_{xx}+1)^2 u + \eps^2  u - u^3,\qquad x>0,\qquad \partial_{xx}u+u=0,\qquad \partial_{xxx}u+\partial_x u=0,\qquad x=0. 
\]
These somewhat specific boundary conditions in fact arise naturally as boundary conditions when  minimizing the energy  
\[
\mathcal{E}(u)=\int_\R \left(\left((\partial_{xx}+1)u\right)^2-\eps^2 u^2 +\frac{1}{2}u^4\right)\rmd x.
\]
over all of $H^2(\R_+)$. 

\begin{figure}[h!]
    \centering
    \includegraphics[width=0.465\textwidth]{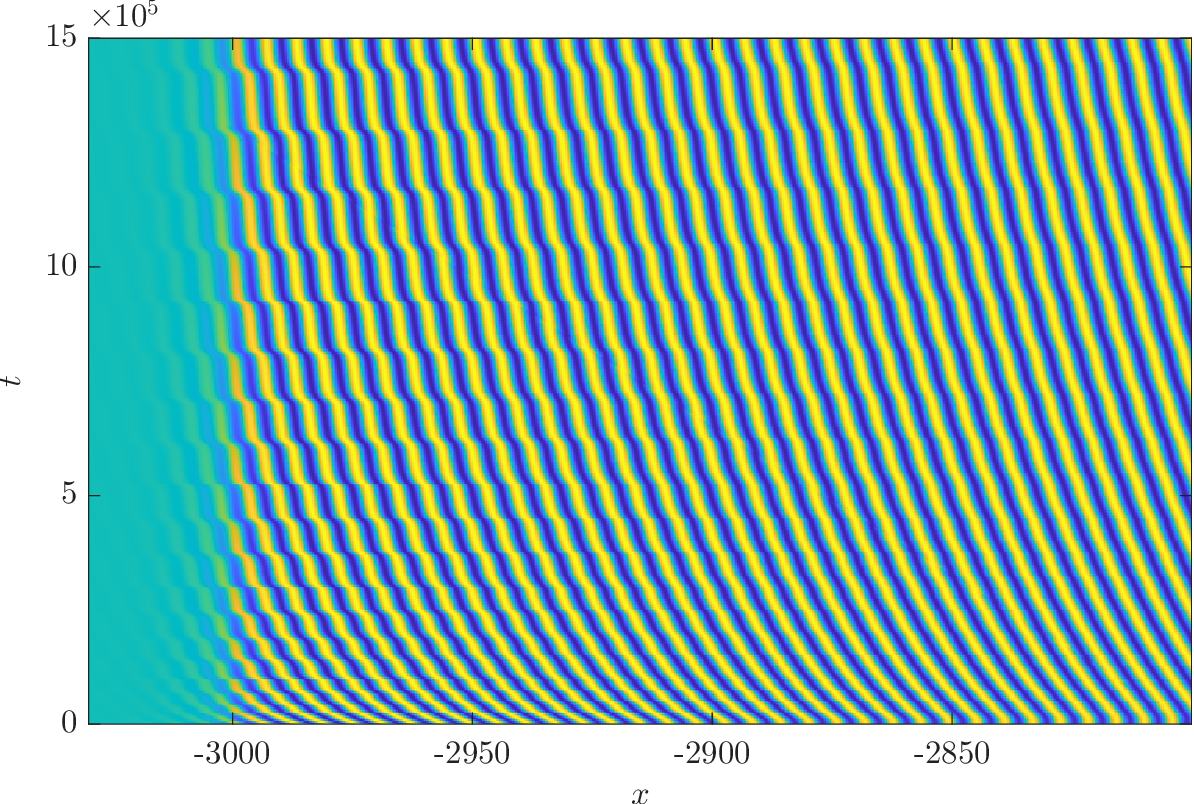}\hfill  \includegraphics[width=0.49\textwidth]{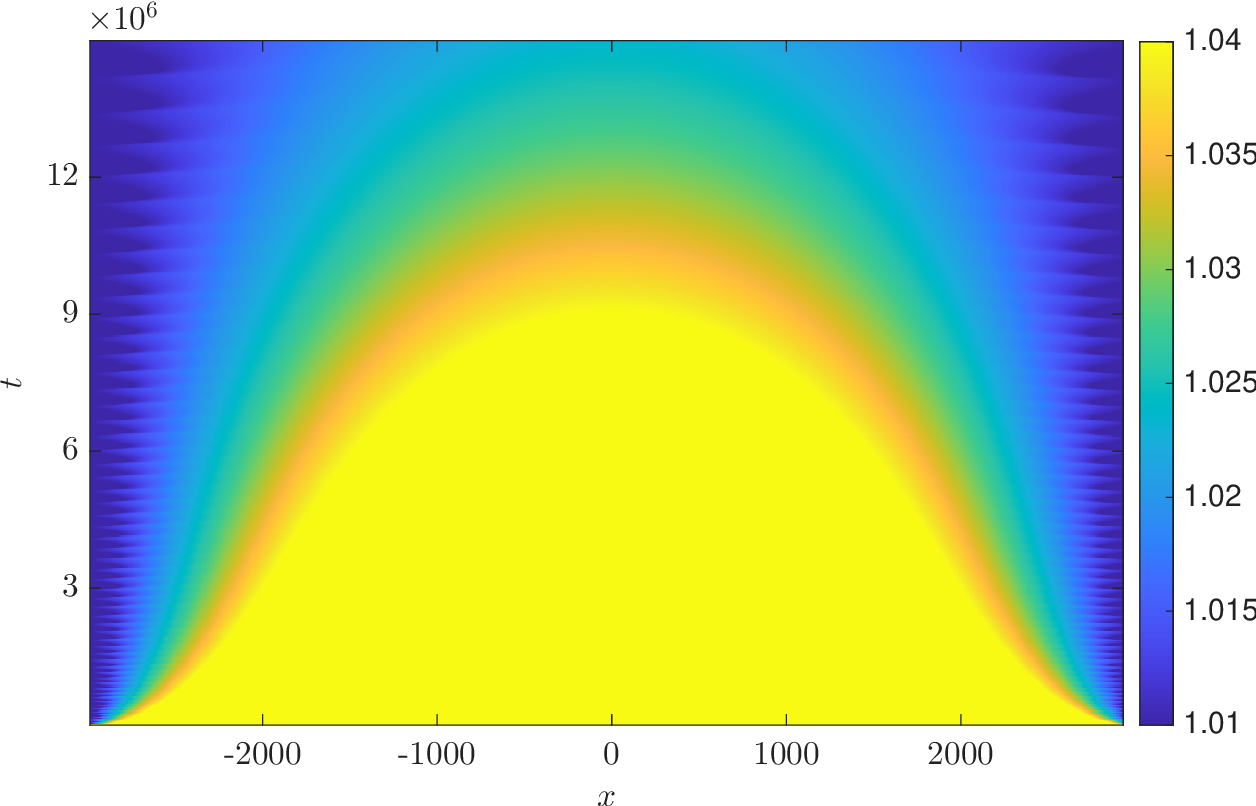}\hfill  \\
    \includegraphics[width=1\textwidth]{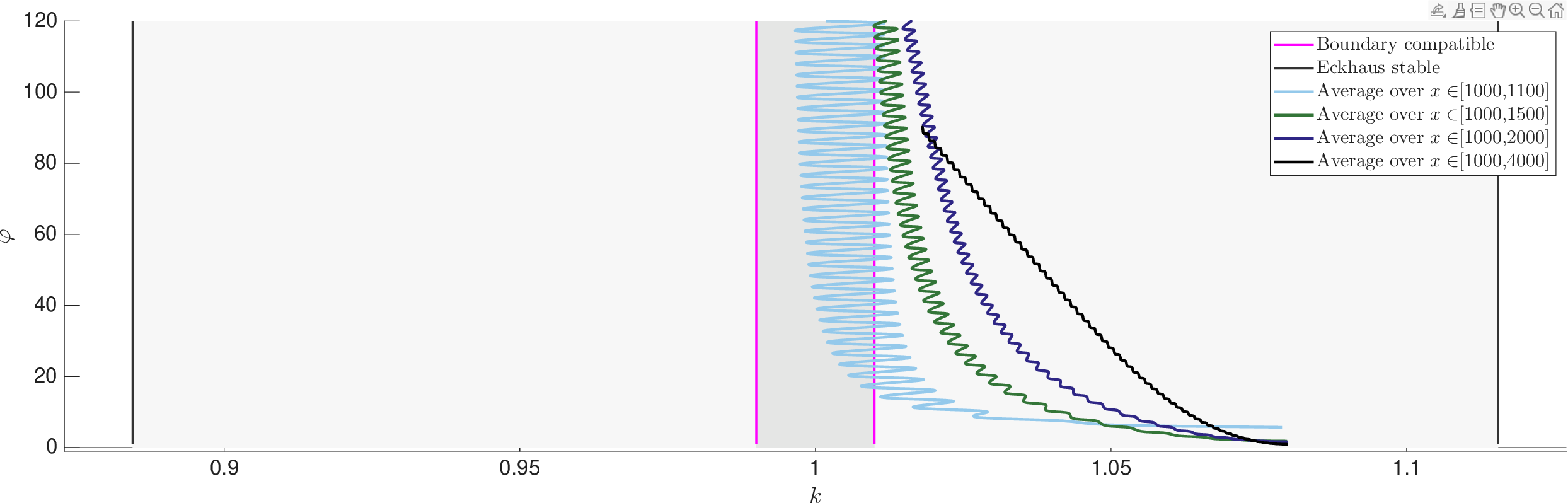} 
    \caption{Direct simulations of the Swift-Hohenberg equation~\eqref{e:SH} with parameter jump and initial conditions as described in the text. Part of the space-time region is shown (top left) to demonstrate the snapping at the boundary, eliminating stripes, and thus reducing the overall wavenumber. Snapping events are spaced at a scaling $t_n\sim n^2$, slowing down as time progresses. We illustrate both snapping and the diffusive spreading of a corrected wavenumber into the bulk of the domain by plotting the derivative of the effective phase (see text) $k=\varphi_x$ (top right). The convergence to the marginally stable wavenumber is illustrated by plotting  average phase and wavenumber in a ``phase plane'' as functions of time (bottom). Solutions converge to the maximal wavenumber $k_\mathrm{max}\sim 1+\eps^2/16=1.01$ as time progresses. Convergence is faster when averaging over a smaller neighborhood of the boundary as opposed to the entire domain, but snapping events have a more significant impact on the average in small regions.   }\label{f:sh}
\end{figure}
\medskip\textbf{Selection of marginal states.}
We studied numerically the prototypical example of the Swift-Hohenberg equation in a large domain, 
\begin{equation}\label{e:SH}
    \partial_t u = - (\partial_{xx}+1)^2 u + \eps^2 \mu_L(x) u - u^3,\qquad x\in (-L,L),\quad L=4000, 
\end{equation}
with pattern forming region confined to the center of the domain, 
\[
 \mu_L(x)=\left\{\begin{array}{ll} 
 1,&|x|<3L/4,\\
 -1,&|x|>3L/4,
 \end{array}\right.
\]
periodic boundary conditions, and driving parameter $\eps=0.4$. This leads to pattern formation in the central region $|x|<3000$ and values $u$ close to zero in the complement. As initial condition we chose $u_0(x)=\frac{2\eps}{\sqrt{3}}\sqrt{1-4 \kappa^2}\cos((1+\eps\kappa)x)$ with $\kappa=0.2$, thus detuning from the center of the band of wavenumbers. We simulated the equation using a smoothed version of the sign function and $2^{16}$ Fourier modes in a spectral discretization, and an exponential time differencing scheme \textsc{ETDRK4} with $dt=0.1$. We simulated the system until a final time of $1.5\cdot 10^6$. We then extracted an effective phase $\varphi$ from the solution by computing $\arctan(u_x/u)$, which we smoothed with a Gaussian blur of width $2\pi$.  

The simulations clearly show how the solution converges in a very weak sense to the equilibrium with a wavenumber that is maximal among wavenumbers compatible with the parameter step; see Figure~\ref{f:sh}. The fact that this marginal equilibrium is unstable is also reflected in the repeated snapping of the solutions. The instability can be thought of as comparable to the instability of a saddle-node equilibrium, stable from above, but unstable below. The equilibrium here is in fact of saddle-node type made precise in~\cite{MR4135828} and snapping events follow the one-dimensional unstable manifold. 

Our main result establishes this weak convergence of wavenumbers to the marginally stable wavenumber at the edge of the compatible band, in a reduced, somewhat universal model that we shall introduce next. 

\begin{figure}[h!]
    \centering
    \includegraphics[width=0.48\textwidth]{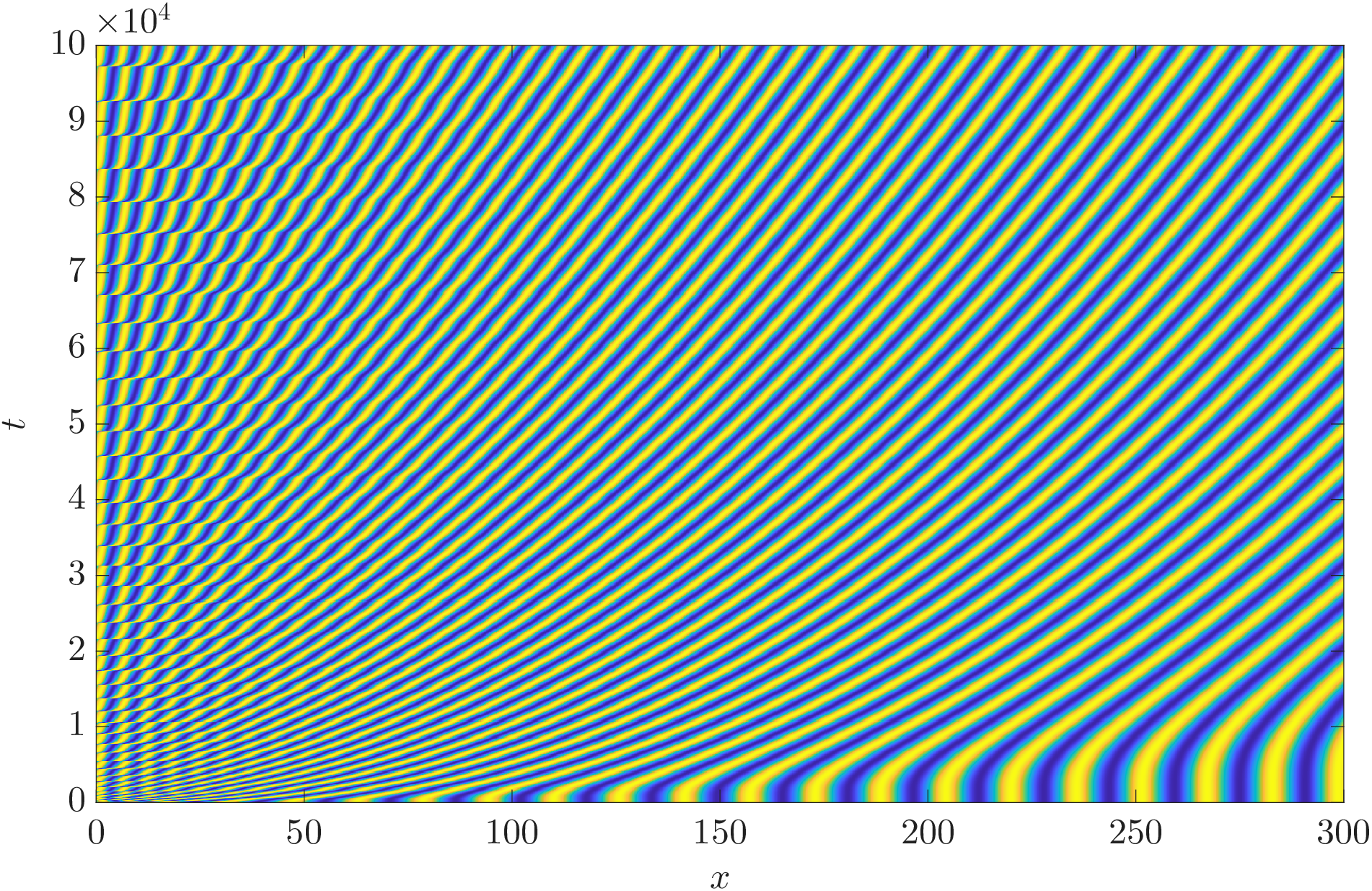}\hfill   \includegraphics[width=0.48\textwidth]{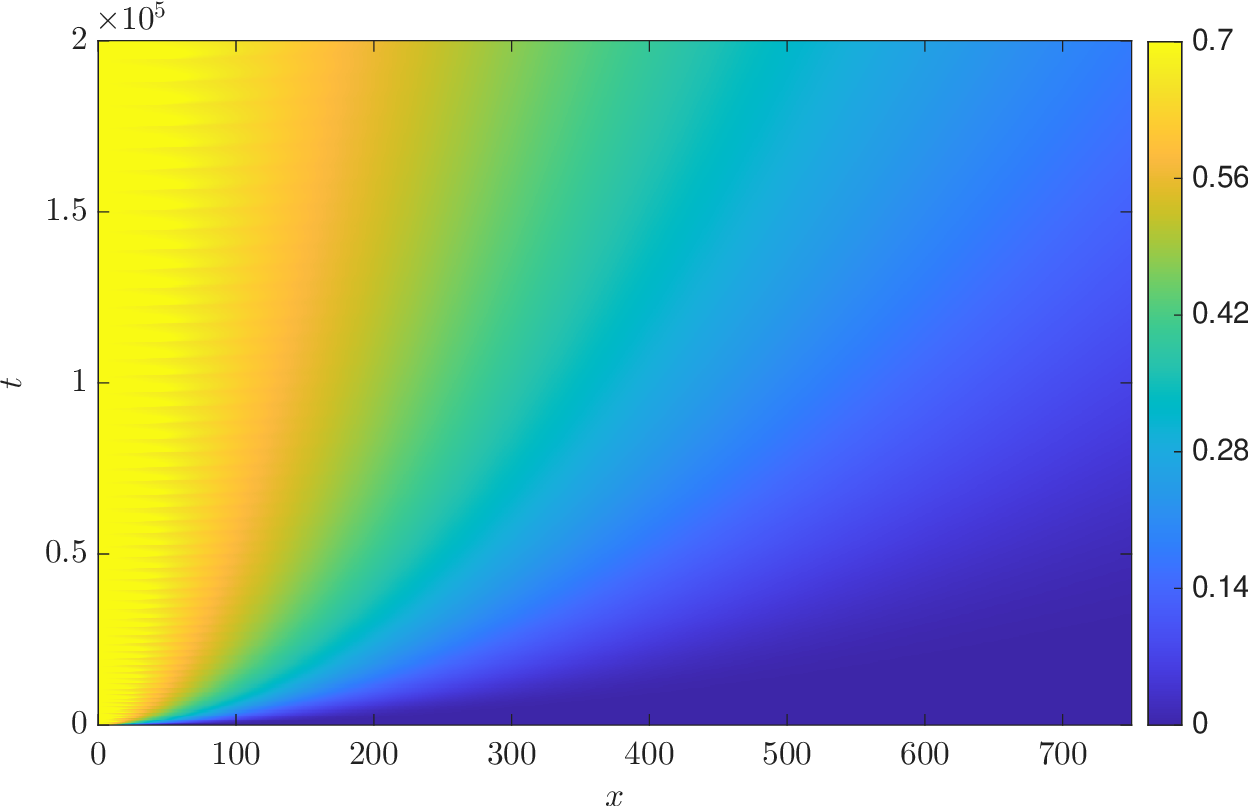} \hfill \\
    \includegraphics[width=1\textwidth]{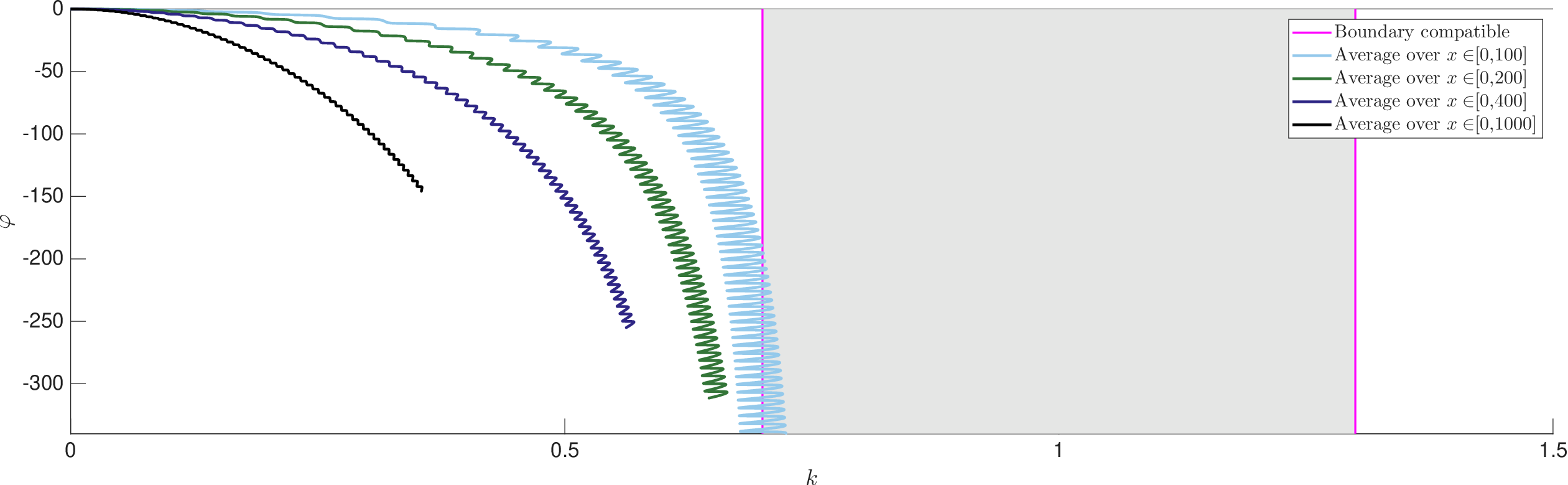} \\
    \caption{Direct numerical simulations of the phase diffusion equation~\eqref{PD} for $\mu=0.3$ and $\varphi_0(x)=0$ with Neumann boundary conditions at $x=3000$. The space-time plot (top left) shows $\cos(0.4 x + \varphi(t,x))$ in part of the simulated space-time region and demonstrates the same snapping behavior at the boundary that was observed in Figure~\ref{f:sh}. Also shown (top right) is the space-time plot of $\partial_x \varphi$ which reproduces half of the phase plot in Figure~\ref{f:sh}. The convergence to the marginally stable wavenumber is illustrated via a ``phase plane'' plot (bottom), showing how average slopes approach the left edge of the range of admissible slopes, through repeated snapping events at which the phase at the boundary decreases by $2\pi$. We again see that convergence is faster and snapping events are more pronounced when looking at  averages over smaller regions.}
    \label{f:pd}
\end{figure}

\medskip\textbf{The phase-diffusion approximation.} 
On the real line, one can modulate the phase in $u_\mathrm{per}(kx-\varphi;k)$ slowly with $\varphi=\varphi(\delta^2 t,\delta x)$ and obtain, adjusting the wavenumber accordingly to improve the approximation, at leading order a modulation equation for the phase, 
\begin{equation}\label{e:pd}
    \partial_t\varphi=d_\mathrm{eff}\partial_{xx}\varphi;
\end{equation}
see for instance~\cite{MR2029750} for a derivation and justification in the Ginzburg-Landau equation, \cite{MR2507940} for a more general setting, and \cite{MR763474} for the original derivation in a two-dimensional setting.
Here $d_\mathrm{eff}$ can be obtained from a linear stability analysis of the pattern with wavenumber $k$ near which the modulation is carried out. 

A derivation of this equation in the presence of boundary conditions is not known. Nevertheless, positing~\eqref{e:pd} as a model for pattern dynamics by thinking of the pattern as given by $\cos(\varphi(t,x))$, we would infer that the correct displacement-strain relation in this case is $\varphi_x=1+\frac{\eps^2}{16}\sin(2\varphi)$. The diffusion equation~\eqref{e:pd} together with  this boundary condition, is the equation that we shall study in this paper. After some simple scalings that reduce to $2\pi$- rather than $\pi$-periodicity, scale the amplitude of the $\sin$-nonlinearity, and normalize the diffusion coefficient to $1$, the equation is given by
\begin{equation}\label{PD}
  \left\{\begin{array}{ll}
    \partial_t \varphi = \partial_{xx} \varphi, &t,x>0,\\
    \partial_x \varphi(t,0) = 1 + \mu \sin(\varphi(t,0)), & t>0,  \\
    \varphi(0,x) = \varphi_0(x), & x>0,
\end{array}\right.
\end{equation}
where $\varphi \in \R/(2\pi\Z)$ and $0< \mu < 1$. We immediately see how equilibria are given by affine lines of the form $\varphi(x)=ax+b$ where $a=1+\mu \sin(b)$, thus mimicking the displacement-strain relation in the Swift-Hohenberg equation.

In summary, (local) slopes $\varphi_x$ of solutions $\varphi$ to~\eqref{PD} correspond to (local)  wavenumbers of solutions to the Swift-Hohenberg equation~\eqref{e:SH0}. For equilibria, the boundary condition in~\eqref{PD} confines slopes, and thus wavenumbers in~\eqref{e:SH0}, to $[1-\mu, 1+\mu]$.

\medskip\textbf{Stability, instability, and marginal stability.}
In the phase-diffusion approximation, stability and instability of equilibria can be readily seen explicitly. Indeed, linearizing at a profile $\varphi_*(x)=ax+b$ satisfying $a=1+\mu \sin(b)$, we find 
\[
\partial_t\varphi = \partial_{xx}\varphi, \quad x>0,\qquad \partial_x\varphi = (\mu\cos b)\varphi, \quad x=0.
\]
This equation exhibits exponentially growing solutions  precisely when there exists a solution of the form $\varphi(t,x)=\rme^{\lambda t}\varphi_0(x)$  associated with a localized eigenfunction $\varphi_0(x)=\rme^{\gamma x}$, $\Re\gamma<0$, to a positive eigenvalue $\lambda=\gamma^2>0$. The boundary conditions then imply that $\gamma=\mu\cos b$, which is negative when   $b\in (\pi/2,3\pi/2)$. 
%
%
Extremal wavenumbers $\partial_x \varphi=1\pm\mu$ correspond to $b\in \{-\pi/2,\pi/2\}$ and are therefore marginally stable.

\medskip\textbf{Main result: selection of marginally stable slopes.}
Our main result states that initial conditions with slopes $\partial_x\varphi_0(x)<1-\mu$, corresponding to wavenumbers that are too small to be compatible with the boundary condition, will evolve so that the slope approaches the slope of the extremal equilibrium in a weak sense. 
\begin{theorem}\label{t:main}
    Let $\varphi(t,x)$ be a solution to~\eqref{PD} with initial condition $\varphi_0(x)=\alpha x + \beta$. 
    \begin{enumerate}
        \item If $\alpha<1-\mu$, then for all $\varepsilon>0$ there exists $T>0$ such that 
            \begin{equation*}
                \Big|\partial_x\varphi\big(t,\sqrt{\varepsilon t}\big)-\big(1-\mu\big)\Big|\lesssim\sqrt{\eps} \quad \text{for all } t>T.
            \end{equation*}
        \item If $\alpha>1+\mu$, then for all $\varepsilon>0$ there exists $T>0$ such that 
            \begin{equation*}
                \Big|\partial_x\varphi\big(t,\sqrt{\varepsilon t}\big)-\big(1+\mu\big)\Big|\lesssim\sqrt{\eps} \quad \text{for all } t>T.
            \end{equation*}
    \end{enumerate}
\end{theorem}
We illustrate the result in direct simulations, shown in Figure~\ref{f:pd}, which show a remarkable agreement with the dynamics in the Swift-Hohenberg equation, Figure~\ref{f:sh}. 

\begin{figure}
    \includegraphics[width=\linewidth]{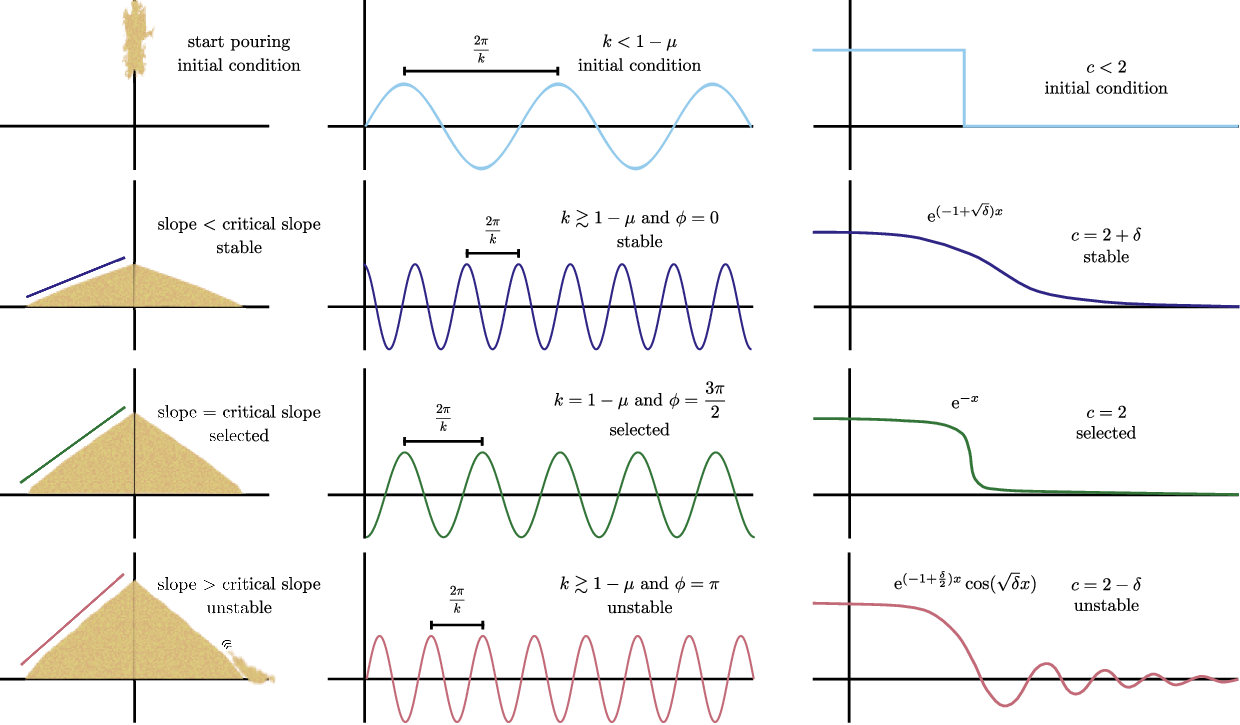}
    \caption{Schematic illustration suggesting similarities between the selection of marginally stable states in the sand pile (left), front invasion (right), and patterns with boundary (center, present work); see text for details.    }
    \label{f:soc}
\end{figure}

\medskip\textbf{Self-organized criticality.}
We conclude this introduction with some comments on the selection of a marginally stable state; see Figure~\ref{f:soc}. The idea that dynamical systems often select such critical state was laid out and corroborated in~\cite{PhysRevLett.59.381}. The resulting fragility of the final state, being only marginally stable and hence susceptible to perturbations, leads to wide-ranging predictions of statistics in complex systems~\cite{PhysRevLett.59.381,OByrne2022}. However, specific instances where the convergence to critical, marginally stable states has been established are rare. Besides the original Abelian sandpile model~\cite{DHAR19994}, we mention here the selection of fronts in the context of propagation into unstable state; see~\cite{avery2026selectionmechanismsinvasion,vanSaarloosReview} for reviews. In the sandpile, stable slopes exist for a range of angles, yet the steepest and least stable angle is selected. Similarly, among all invasion fronts in say the Fisher-KPP equation, the steepest front, that is, the front with the largest exponential decay rate, is selected when initial conditions possess an even steeper decay rate.  In our case, the wavenumber of patterns plays the role of a slope or exponential decay rate in these two examples, and initial conditions with a smaller (or larger) wavenumber converge to an extremal wavenumber with a marginally stable profile. 

\medskip\textbf{Outline.} We collect the relevant background on comparison principles and obtain upper bounds on solutions when $\varphi_0=\beta$ in \S\ref{s:3}. We construct subsolutions in the same setting using boundary data that mimics the snapping at the boundary in \S\ref{s:4}. We derive estimates for the solution from boundary data and wrap up the proof of Theorem~\ref{t:main} in \S\ref{s:5}, and conclude by presenting a summary with further questions in \S\ref{s:dis}.

\section{Comparison principles and supersolutions}\label{s:3}
The proof of Theorem~\ref{t:main} relies on bounding the solution by  sub- and  supersolutions. 

\medskip\textbf{Comparison principles.}
We will begin with a definition of these terms.
\begin{definition}[Sub- and supersolutions]\label{d:sssol}
    Let $T>0$ and define $D_T = (0,T] \times (0,\infty)$. A function $\varphi \in C(\overline{D_T}) \cap C^2(D_T)$ is called a supersolution if
        \begin{equation}\label{d:sup-crit}
            \begin{cases} 
                \partial_t\varphi \geq \partial_{xx} \varphi, & (t,x) \in D_T,\\    
                \partial_x\varphi(t,0) \leq 1+\mu\sin(\varphi(t,0)), & t \in (0,T],\\
                \varphi(0,x) \geq \varphi_0(x), & x \in (0,\infty).
            \end{cases}
        \end{equation}
    Similarly a function $\varphi \in C(\overline{D_T}) \cap C^2(D_T)$ is called a subsolution if it satisfies all of the reversed inequalities above.
\end{definition}
Sub- and supersolutions yield barriers for actual solutions by the following standard result; see~\cite[Theorem 4.1.1]{PaoNonlinear}.
\begin{proposition}[Parabolic comparison]\label{t:cc}
    Let $T>0$ be fixed and $\varphi$ be the unique solution to~\eqref{PD}. Suppose $\varphi_{\mathrm{sub}}$ and $\varphi_{\mathrm{sup}}$ are sub- and supersolutions in $D_T$, as defined in Definition~\ref{d:sssol}. Then
    \[ 
        \varphi_{\mathrm{sub}} \leq \varphi \leq \varphi_{\mathrm{sup}} \quad \text{in} \quad D_T.
    \]
\end{proposition}

\medskip\textbf{Supersolutions.}
To find a supersolution, we solve the following boundary value problem
\begin{equation*}
\begin{cases}
    \partial_t \varphi=\partial_{xx} \varphi, & t,x>0,\\
    \partial_x \varphi(t,0)=1-\mu,  &t>0,\\
    \varphi(0,x)=\beta,   &x>0.
\end{cases}
\end{equation*}
The explicit solution is given by
\begin{equation}\label{e:supsoln}
    \varphi_{\mathrm{sup}}(t,x) = \beta+(1-\mu )\,x-\frac{2(1-\mu)\sqrt{t}}{\sqrt{\pi}}e^{-x^2/4t}-(1-\mu)\,x\erf\Big(\frac{x}{2\sqrt{t}}\Big),
\end{equation}
where
$$\erf(\xi) = \frac{2}{\sqrt{\pi}} \int_0^\xi e^{-z^2}\,dz.$$

\begin{proposition}\label{p:sup}
    The function $\varphi_{\mathrm{sup}}$ is a supersolution for~\eqref{PD}, when $\varphi_0=\beta$, in $D_T$ for all $T>0$.
\end{proposition}
\begin{proof}
    By construction, $\varphi_{\mathrm{sup}}$ trivially satisfies the first and third inequalities in~\eqref{d:sup-crit}. The second is satisfied since
    $$\partial_x \varphi_{\mathrm{sup}}(t,0)  = 1-\mu \leq 1+\mu \sin(\varphi_{\mathrm{sup}}(t,0)) \quad \forall \, t>0.$$
    As a consequence, $\varphi_{\mathrm{sup}}$ is a supersolution in $D_T$ for all $T>0$.
\end{proof}
We also note here that
\begin{equation}\label{e:supderiv}
    \partial_x \varphi_{\mathrm{sup}}(t,x) = (1-\mu)\Big(1-\erf\big(\frac{x}{2\sqrt{t}}\big)\Big).
\end{equation}

\section{Boundary snapping and subsolutions}\label{s:4}
Finding a useful subsolution proves to be more difficult. To start, consider the following boundary value problem
\begin{equation}\label{e:rec}
    \begin{cases}
    \partial_t \varphi=\partial_{xx}\varphi, & t,x>0,\\
    \varphi(t,0)=\varphi_*(t), & t>0,\\
    \varphi(0,x)=\varphi_0(x), & x>0.
\end{cases}
\end{equation}
The explicit solution is given by
\begin{equation}\label{e:recsoln}
    \varphi(t,x)=\frac{1}{\sqrt{4\pi t}}\int_0^\infty\bigg(e^{-(x-y)^2/4t}-e^{-(x+y)^2/4t}\bigg)\varphi_0(y) dy+\int_0^t\frac{x}{\sqrt{4\pi s^3}}e^{-x^2/4s}\varphi^*(t-s) ds.
\end{equation}
If we knew the explicit boundary trace for a solution to~\eqref{PD}, then we could substitute it into~\eqref{e:rec} in place of $\varphi_*(t)$ and, by uniqueness, we would obtain the same solution. Thus if we choose $\varphi_*(t)$ so that it sits below the boundary trace, then we would hope to obtain a subsolution to~\eqref{PD}. With this in mind, we consider the following boundary value problem
\begin{equation}\label{e:sub-rec}
    \begin{cases}
    \partial_t \varphi=\partial_{xx} \varphi, & t,x>0,\\
    \varphi(t,0)=\varphi_*^\varepsilon(t), & t>0,\\
    \varphi(0,x)=\beta, & x>0,
\end{cases}
\end{equation}
where
\begin{equation*}
    \varphi_*^\varepsilon(t)=-2\pi \Big\lfloor \frac{1-\mu+\varepsilon}{\pi^{\frac{3}{2}}}\sqrt{t} \Big\rfloor -\frac{\pi}{2}-2\pi M,
\end{equation*}
for $\varepsilon>0$, $M \in \Z_{>0}$, and $\lfloor x \rfloor = \floor(x)=\max\{j\in\Z\,|\,j\leq x\}$. A plot comparing this ansatz with the actual boundary trace is shown in Figure~\ref{f:subsoln}.

The explicit solution to~\eqref{e:sub-rec} is then given by
\begin{equation}\label{e:subsoln}
    \varphi^\varepsilon_{\mathrm{sub}}(t,x) = \varphi^\varepsilon_*(t)+\Big(\frac{\pi}{2} + 2 \pi M + \beta\Big) \erf\Big(\frac{x}{2\sqrt{t}}\Big)+2\pi\sum_{k=1}^\infty\erf\Big(\frac{x}{2\sqrt{t-s_k^\varepsilon}}\Big)\chi_{\{t\geq s_k^\varepsilon\}}(t),
\end{equation}
where $\chi$ is the indicator function and
\begin{equation}\label{e:skdef}
    s_k^\varepsilon = \frac{k^2 \pi^3}{(1-\mu+\varepsilon)^2}.
\end{equation}
We note here that the $s_k^\varepsilon$ are such that if $s_k^\varepsilon \leq t < s_{k+1}^\varepsilon$ then $\Big\lfloor \dfrac{1-\mu+\varepsilon}{\pi^{\frac{3}{2}}}\sqrt{t} \Big\rfloor = k.$
\begin{figure}
    \centering
    \includegraphics[width=0.48\textwidth]{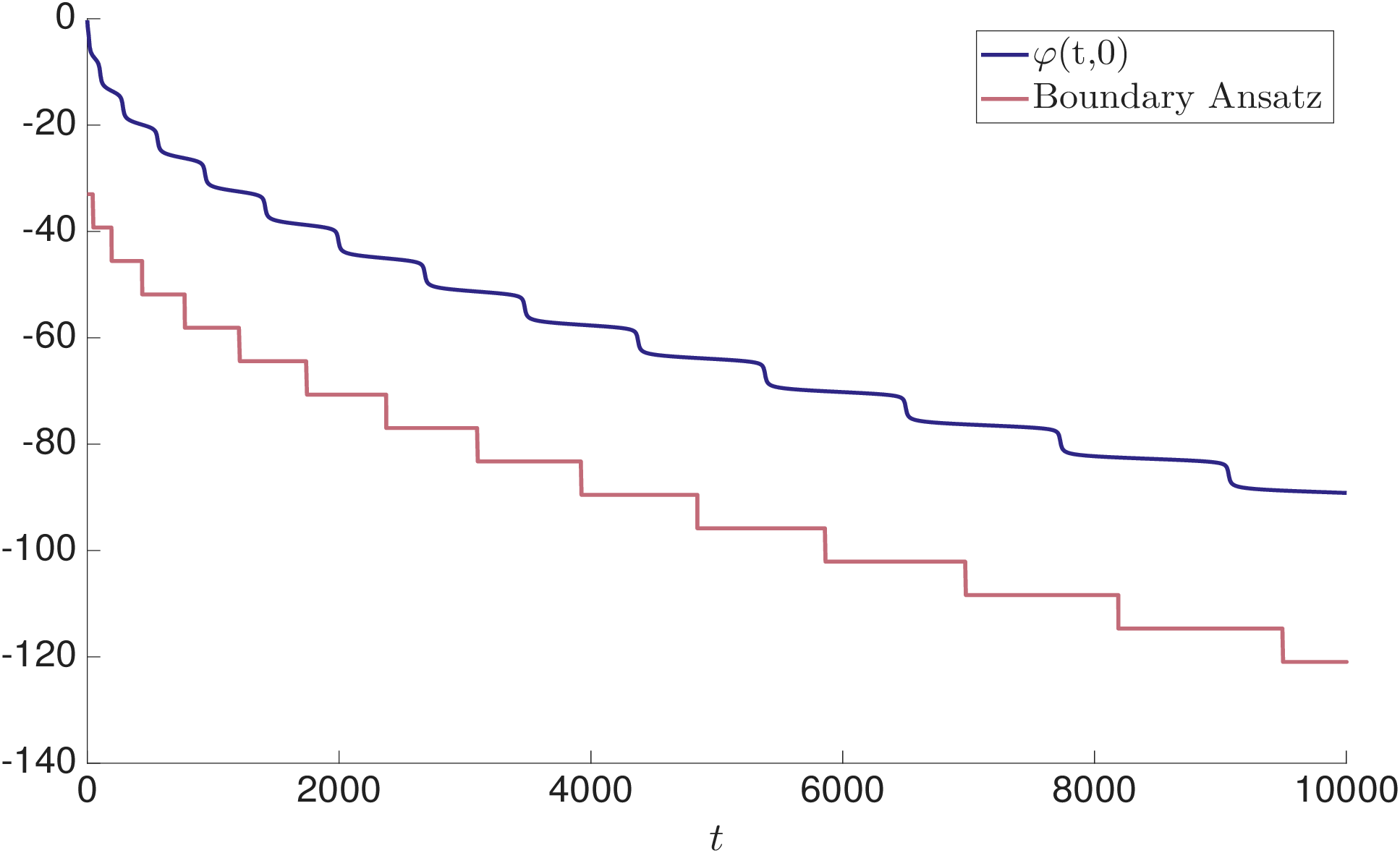} \hfill 
    \includegraphics[width=0.48\textwidth]{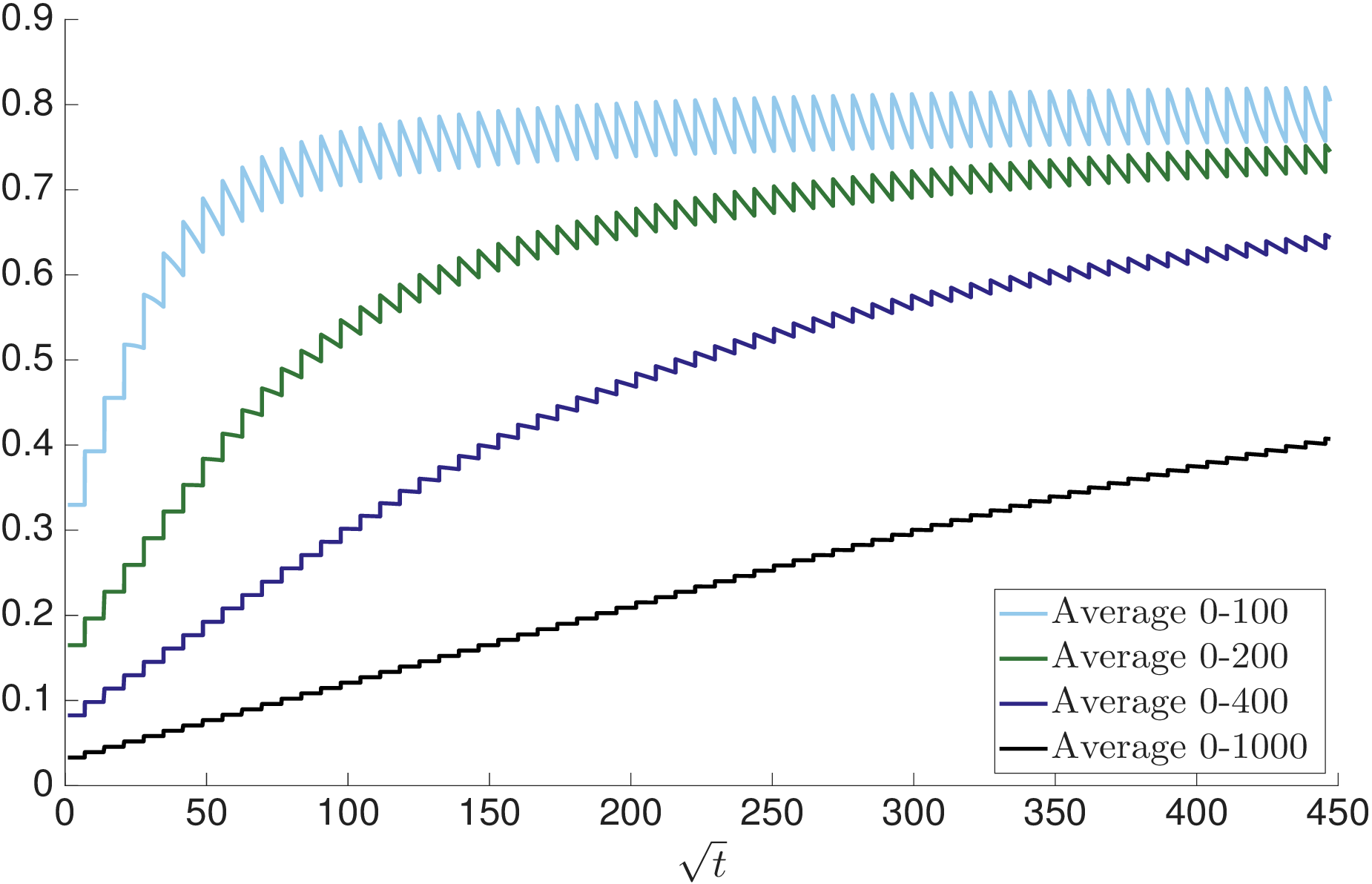}
    \caption{Plots using the subsolution~\eqref{e:subsoln} when $\mu=0.3$, $\varepsilon=0.1$, $M=5$, and $\beta = 0$. The boundary ansatz $\varphi_*^\varepsilon(t)$ is plotted against the boundary trace for a direct numerical simulation of~\eqref{PD} (left) with $\mu=0.3$ and $\varphi_0=0$, motivating the choice of boundary ansatz made. The averages of $\partial_x \varphi^\varepsilon_\mathrm{sub}$ for $x$ in four different intervals are plotted against $\sqrt{t}$ (right). We observe that the subsolution exhibits the same $\sqrt{t}$-periodic, snapping behavior seen in Figure~\ref{f:sh} and~\ref{f:pd}. One difference is that the averages here do not converge to $1-\mu=0.7$. Instead it can be shown, using a similar argument seen in the proof of Lemma~\ref{l:sum}, that they will converge in some sense to $1-\mu +\varepsilon = 0.8$.}
    \label{f:subsoln}
\end{figure}
While this function  mimics  the snapping behavior observed in Figs.~\ref{f:pd} and~\ref{f:subsoln}, the definition of a subsolution  requires continuity of the boundary data at $x=0$, so we cannot immediately apply Proposition~\ref{t:cc}. To rectify this, we replace our ansatz $\varphi_*^\varepsilon$ with a continuous approximation,
\begin{equation}\label{ContAnsatz}
\Tilde{\varphi}_\delta^\varepsilon(t)=\varphi_*^\varepsilon(t)+\tilde{\psi}_\delta^\varepsilon(t),
\end{equation}
where
\begin{equation*}
    \tilde{\psi}_{\delta}^\varepsilon(t)=\begin{cases}
        -\dfrac{2\pi}{\delta}(t-s_k^\varepsilon)+2\pi & s_k^\varepsilon\leq t\leq s_k^\varepsilon+\delta\quad \text{for all }k\in\N, \\
        0 & \text{otherwise,}
    \end{cases}
\end{equation*}
for some $0 < \delta < s_1^\varepsilon$. By construction $\tilde{\varphi}_\delta^\varepsilon$ is continuous and $\varphi_*^\varepsilon(t) \leq \tilde \varphi_\delta^\varepsilon(t)$. Replacing $\varphi_*^\varepsilon$ with $\tilde \varphi_\delta^\varepsilon$ in~\eqref{e:sub-rec}, we can again explicitly solve. We denote the resulting solution in the full domain by $\varphi_\delta^\varepsilon$. 
If $s_n^\varepsilon \leq t \leq s_{n}^\varepsilon + \delta$ then $\varphi_\delta^\varepsilon$ is given by
\begin{align}\label{e:sub-delta1}
    \everymath{\displaystyle}
    \begin{array}{cl}
    \varphi^\varepsilon_\delta(t,x) &= \varphi^\varepsilon_{\mathrm{sub}}(t,x)+\sum_{k=1}^{n-1}\Bigg[-\frac{2x\sqrt{\pi}}{\delta}e^{-x^2/4(t-s_k^\varepsilon-\delta)}\sqrt{t-s_k^\varepsilon-\delta}+\frac{2x\sqrt{\pi}}{\delta}e^{-x^2/4(t-s_k^\varepsilon)}\sqrt{t-s_k^\varepsilon}\\
    &\quad +\bigg(-\frac{2\pi}{\delta}(t-s_k^\varepsilon)+2\pi-\frac{\pi}{\delta}x^2\bigg)\Bigg(\erf\Big(\frac{x}{2\sqrt{t-s_k^\varepsilon-\delta}}\Big)-\erf\Big(\frac{x}{2\sqrt{t-s_k^\varepsilon}}\Big)\Bigg)\Bigg]\\
    &\quad +\frac{2x\sqrt{\pi}}{\delta}e^{-x^2/4(t-s_n^\varepsilon)}\sqrt{t-s_n^\varepsilon}+\bigg(-\frac{2\pi}{\delta}(t-s_n^\varepsilon)+2\pi-\frac{\pi}{\delta}x^2\bigg)\Bigg(1-\erf\Big(\frac{x}{2\sqrt{t-s_n^\varepsilon}}\Big)\Bigg).
    \end{array}
\end{align}
For $s_n^\varepsilon + \delta \leq t < s_{n+1}^\varepsilon$, it is explicitly given by
\begin{align}\label{e:sub-delta2}
    \everymath{\displaystyle}
    \begin{array}{cl}
    \varphi^\varepsilon_\delta(t,x) &= \varphi^\varepsilon_{\mathrm{sub}}(t,x)+\sum_{k=1}^{n}\Bigg[-\frac{2x\sqrt{\pi}}{\delta}e^{-x^2/4(t-s_k^\varepsilon-\delta)}\sqrt{t-s_k^\varepsilon-\delta}+\frac{2x\sqrt{\pi}}{\delta}e^{-x^2/4(t-s_k^\varepsilon)}\sqrt{t-s_k^\varepsilon}\\
    &\quad +\bigg(-\frac{2\pi}{\delta}(t-s_k^\varepsilon)+2\pi-\frac{\pi}{\delta}x^2\bigg)\Bigg(\erf\Big(\frac{x}{2\sqrt{t-s_k^\varepsilon-\delta}}\Big)-\erf\Big(\frac{x}{2\sqrt{t-s_k^\varepsilon}}\Big)\Bigg)\Bigg].
    \end{array}
\end{align}
Before we continue with the proof, note that by~\eqref{e:recsoln} and the fact that $\varphi_*^\varepsilon(t) \leq \tilde \varphi_\delta^\varepsilon(t)$, we immediately can conclude that
\begin{equation*}
    \varphi^\varepsilon_{\mathrm{sub}}(t,x) \leq \varphi_\delta^\varepsilon(t,x) \quad \forall \,\, t,x\geq0.
\end{equation*}
Therefore,  if we can show that $\varphi_\delta^\varepsilon$ is a subsolution, we also have that for $\varphi$ the solution to~\eqref{PD}, 
\begin{equation*}
    \varphi^\varepsilon_{\mathrm{sub}}(t,x) \leq \varphi(t,x) \quad \forall \,\, t,x\geq0.
\end{equation*} 
To prove that $\varphi_\delta^\varepsilon(t,x)$ is a subsolution, we will first need the following lemma.
\begin{lemma}\label{l:sum}
    For all $\varepsilon>0$, we have with $s_k^\varepsilon$ from~\eqref{e:skdef} that
    \begin{equation*}
        \lim_{n \to \infty} 2\sqrt{\pi}\sum_{k=1}^{n}\frac{1}{\sqrt{s_{n+1}^\eps-s_k^\eps}} = 1-\mu+\varepsilon.
    \end{equation*}
\end{lemma}
\begin{proof}
    To start, define $\sigma_k = \sqrt{\dfrac{s_k^\varepsilon}{s_{n+1}^\varepsilon}}$ and $\Delta \sigma = \sigma_{k+1}-\sigma_k = \dfrac{1}{n+1}$. Then we can factor the sum to get
    \begin{equation}\label{e:limst}
        2\sqrt{\pi}\sum_{k=1}^{n}\frac{1}{\sqrt{s_{n+1}^\eps-s_k^\eps}} = \dfrac{2(1-\mu+\varepsilon)}{\pi} \sum_{k=1}^n \dfrac{1}{\sqrt{1-\sigma_k^2}} \Delta \sigma.
    \end{equation}
    Define $f(\sigma)=\dfrac{1}{\sqrt{1-\sigma^2}}.$ This function is non-negative, increasing in $\sigma$, and Riemann integrable on $[0,a]$ for all $a<1$, with
    \begin{equation*}
        \lim_{a\to1^-}\int_0^a\frac{1}{\sqrt{1-\sigma^2}}d\sigma=\frac{\pi}{2}.
    \end{equation*}
    Since $f$ is increasing, we have that
    \begin{equation*}
        f(\sigma_k)\Delta\sigma\le\int_{\sigma_k}^{\sigma_{k+1}}\frac{1}{\sqrt{1-\sigma^2}}d\sigma\le f(\sigma_{k+1})\Delta\sigma.
    \end{equation*}
    Summing over $k$ gives
    \begin{equation*}
        \sum_{k=1}^{n-1}f(\sigma_k)\Delta\sigma\le\int_{\sigma_1}^{\sigma_{n}}\frac{1}{\sqrt{1-\sigma^2}}d\sigma\le \sum_{k=2}^nf(\sigma_{k})\Delta\sigma.
    \end{equation*}
    Then, since $f(\sigma)$ is positive, we find that
    \begin{equation*}
        \int_{\sigma_1}^{\sigma_{n}}\frac{1}{\sqrt{1-\sigma^2}}d\sigma+f(\sigma_n)\Delta\sigma\le \sum_{k=1}^{n}f(\sigma_k)\Delta\sigma\le\int_{\sigma_1}^{\sigma_{n}}\frac{1}{\sqrt{1-\sigma^2}}d\sigma+f(\sigma_1)\Delta\sigma.
    \end{equation*}
    Note that by continuity,
    \begin{equation*}
        \lim_{n\to\infty}f(\sigma_1)\Delta\sigma=\lim_{n\to\infty}f\bigg(\frac{1}{n+1}\bigg)\cdot\frac{1}{n+1}=0,
    \end{equation*}
    and similarly
    \begin{equation*}
        \lim_{n\to\infty}f(\sigma_n)\Delta\sigma=\lim_{n\to\infty}f\Big(\frac{n}{n+1}\Big)\cdot \frac{1}{n+1}=\lim_{n\to\infty}\frac{1}{\sqrt{2n+1}}=0.
    \end{equation*}
    Therefore, the squeeze theorem gives that
    \begin{equation*}
    \lim_{n\to\infty}\sum_{k=1}^nf(\sigma_k)\Delta\sigma=\lim_{n\to\infty}\int_{\sigma_1}^{\sigma_n}f(\sigma)d\sigma=\int_0^1f(\sigma)d\sigma=\frac{\pi}{2}.
    \end{equation*}
    Substituting this into our expression for the limit~\eqref{e:limst} above yields the desired result:
    $$\lim_{n \to \infty} \frac{2(1-\mu+\varepsilon)}{\pi}\sum_{k=1}^n\frac{1}{\sqrt{1-\sigma_k^2}}\Delta \sigma = 1-\mu+\varepsilon.$$
\end{proof}
Using this lemma, we are now able to prove the following result.
\begin{proposition}\label{p:sub-delta}
    For any $\varepsilon>0$, $M=M_\varepsilon$ sufficiently large and $\delta$ sufficiently small, the function $\varphi_\delta^\varepsilon$ from~\eqref{ContAnsatz} is a subsolution for~\eqref{PD}, when $\varphi_0 = \beta$, in $D_T$ for all $T>0$.
\end{proposition}
\begin{proof}
    As in the supersolution case, the first and third reverse inequalities of~\eqref{d:sup-crit} are trivially satisfied. So we will focus on showing the reverse of the second, which is given by
    \begin{equation*}
        \partial_x\varphi^\eps_\delta(t,0)\ge1+\mu\sin(\varphi_*^\eps(t,0)+\tilde \psi_\delta^\eps(t,0)).
    \end{equation*}
    Using the angle addition formula and the fact that $\sin(\varphi_*^\varepsilon(t))=-1$ for all $t>0$, the above inequality reduces to
    \begin{equation}\label{e:cc}
            \partial_x\varphi^\eps_\delta(t,0)\ge1-\mu\cos(\tilde{\psi}_\delta^\eps(t)).
    \end{equation}
    Now we split into two cases. First suppose that $t \in [s_n^\varepsilon+\delta, s_{n+1}^\varepsilon)$ for some $n \in \N$. Then $\tilde \psi_\delta^\varepsilon(t) = 0$ and so~\eqref{e:cc} reduces to
    \begin{equation*}
        \partial_x\varphi^\eps_\delta(t,0)\ge1-\mu.
    \end{equation*}
    Thus all we need to show is that the derivative is greater than the constant $1-\mu$. To start, it follows from~\eqref{e:sub-delta2} that
    \begin{equation*}
        \partial_x\varphi^\eps_\delta(t,0)=
        \Big(\frac{\pi}{2}+2\pi M +\beta \Big)\frac{1}{\sqrt{\pi t}}+4\sqrt{\pi}\sum_{k=1}^{n}\frac{1}{\sqrt{t-s_k^\eps}+\sqrt{t-s_k^\eps-\delta}}.
    \end{equation*}
    We can immediately bound the right hand side by below to obtain
    \begin{equation*}
        \partial_x\varphi^\eps_\delta(t,0)\ge \Big(\frac{\pi}{2}+2\pi M+\beta\Big)\frac{1}{\sqrt{\pi t}}+2\sqrt{\pi}\sum_{k=1}^{n}\frac{1}{\sqrt{s_{n+1}^\eps-s_k^\eps}}.
    \end{equation*}
    By Lemma~\ref{l:sum}, there exists an $N_1 \in N$ such that
    \begin{equation*}
        2\sqrt{\pi}\sum_{k=1}^{n}\frac{1}{\sqrt{s_{n+1}^\eps-s_k^\eps}}> 1-\mu+\frac{\eps}{2},
    \end{equation*}
    for all $n \geq N_1$.
    Thus~\eqref{e:cc} holds for all $s_n^\varepsilon +\delta \leq t < s_{n+1}^\varepsilon$ for $n \geq N_1$. So if we can guarantee that
    $$\Big(\frac{\pi}{2}+2\pi M+\beta\Big)\frac{1}{\sqrt{\pi t}}\geq1-\mu+\frac{\varepsilon}{2},$$ 
    for all $t \leq s_{N_1}^\varepsilon$, we are done. To do so, we choose
    $$M \geq M_1 \coloneq \dfrac{(1-\mu+\frac{\varepsilon}{2})\sqrt{\pi s_{N_1}^\varepsilon}-\frac{\pi}{2}-\beta}{2\pi}=\dfrac{\sqrt{\pi^4 N_1^2 }-\frac{\pi}{2}-\beta}{2\pi}=\dfrac{\pi^2 N_1-\frac{\pi}{2}-\beta}{2\pi}.$$

    Next suppose that $t\in[s_n^\eps,s_{n}^\eps+\delta]$ for some $n \in \N$. Then~\eqref{e:sub-delta1} reduces to 
    \begin{equation*}
        \partial_x\varphi^\eps_\delta(t,0)=\Big(\frac{\pi}{2}+2\pi M +\beta \Big)\frac{1}{\sqrt{\pi t}}+4\sqrt{\pi}\sum_{k=1}^{n-1}\frac{1}{\sqrt{t-s_k^\eps}+\sqrt{t-s_k^\eps-\delta}}+\frac{4\sqrt{\pi}}{\delta}\sqrt{t-s_n^\eps}.
    \end{equation*}
    We can again bound this expression from below by 
    \begin{equation*}
        \partial_x\varphi^\eps_\delta(t,0)\ge\Big(\frac{\pi}{2}+2\pi M+\beta\Big)\frac{1}{\sqrt{\pi t}}+2\sqrt{\pi}\sum_{k=1}^{n-1}\frac{1}{\sqrt{s_{n+1}^\eps-s_k^\eps}}+\frac{4\sqrt{\pi}}{\delta}\sqrt{t-s_n^\eps}.
    \end{equation*}
    The sum is quite close to the sum from Lemma~\ref{l:sum}, and a near identical argument shows that this sum converges to $1-\mu+\eps$. So there exists an $N_2\in \mathbb{N}$ such that 
    \begin{equation*}
        2\sqrt{\pi}\sum_{k=1}^{n-1}\frac{1}{\sqrt{s_{n+1}^\eps-s_k^\eps}}\ge 1-\mu+\frac{\varepsilon}{2},
    \end{equation*}
    for all $n\ge N_2$. By the same argument, if we take
    $$M \geq M_2 \coloneq \dfrac{\pi^2 N_2-\frac{\pi}{2}-\beta}{2\pi},$$
    then
    $$\partial_x\varphi_\delta^\varepsilon(t,0) \geq 1-\mu +\frac{\varepsilon}{2} + \frac{4 \sqrt{\pi}}{\delta} \sqrt{t-s_n^\varepsilon},$$
    for all $s_n^\varepsilon \leq t \leq s_{n}^\varepsilon + \delta$ and $n \in \N$. We note here that by choosing such an $M$ we also have that~\eqref{e:cc} also holds for all $s_0^\varepsilon = 0 < t < s_1^\varepsilon$. Therefore in order to show that~\eqref{e:cc} holds, we must show that
    $$1-\mu +\frac{\varepsilon}{2} + \frac{4 \sqrt{\pi}}{\delta} \sqrt{t-s_n^\varepsilon} \geq 1- \mu\cos(\tilde{\psi}^\varepsilon_\delta(t)) \quad \Longleftrightarrow \quad \frac{\varepsilon}{2} + \frac{4 \sqrt{\pi}}{\delta} \sqrt{t-s_n^\varepsilon} - \mu \Big(1-\cos\Big(\frac{2 \pi}{\delta} (t-s_n^\varepsilon)\Big)\Big) \geq 0, $$
    for all $s_n^\varepsilon \leq t \leq s_{n}^\varepsilon + \delta$ and $n \in \N$. For ease of notation let $y = t-s_n^\varepsilon$, then this reduces to
    $$D(y)=\frac{\varepsilon}{2} + \frac{4 \sqrt{\pi}}{\delta} \sqrt{y} - \mu \Big(1-\cos\Big(\frac{2 \pi}{\delta} y\Big)\Big) \geq 0,$$
    for all $0 \leq y \leq \delta$. Note that $D(0)=\dfrac{\eps}{2}>0$, so if $D$ never changes sign then the above inequality holds for all $y\in[0,\delta]$. With this in mind, we will show that $D$ is an increasing function. Taking a derivative we find that
    \begin{equation*}
        D'(y)=\frac{2\sqrt{\pi}}{\delta\sqrt{y}}-\frac{2\mu\pi}{\delta}\sin\Big(\frac{2\pi y}{\delta}\Big)=\frac{2\sqrt{\pi}}{\delta}\bigg(\frac{1}{\sqrt{y}}-\mu\sqrt{\pi}\sin\Big(\frac{2\pi y}{\delta}\Big)\bigg).
    \end{equation*}
    Note that $D'\big(\dfrac{\delta}{2}\big) = \dfrac{\sqrt{8\pi}}{\delta^{3/2}} >0$, and
    \begin{equation}\label{e:delta-cond}
        D'(y)=0 \quad \iff \quad 1=\mu\sqrt{\pi y}\sin\Big(\frac{2\pi y}{\delta}\Big).
    \end{equation}
    Then since $\mu \in (0,1)$ it follows that
    $$\mu\sqrt{\pi y}\sin\Big(\frac{2\pi y}{\delta}\Big) \leq \sqrt{\pi y} \leq \sqrt{\pi \delta},$$
    and so if $\delta < \frac{1}{\pi}$, then~\eqref{e:delta-cond} can never hold. Thus taking $\delta < \frac{1}{\pi}$ gives that $D'(y) > 0$ for all $0 \leq y \leq \delta$. Then by the above discussion for such a $\delta$ and $M$,~\eqref{e:cc} holds for all $s_n^\varepsilon \leq t \leq s_n^\varepsilon+\delta$ and $n \in \N$.
        
    Therefore if we take $M= M_\varepsilon \coloneq \max\{\lceil M_1 \rceil,\lceil M_2 \rceil , 0\}$ and $\delta<\frac{1}{\pi}$, it follows that~\eqref{e:cc} holds for all $t > 0$, and so $\varphi^\varepsilon_\delta$ is a subsolution to the solution of~\eqref{PD} in $D_T$ for all $T>0$.
\end{proof}

Therefore by the above discussion for all $\varepsilon>0$ we also have the desired inequality for the discontinuous $\varphi^\varepsilon_{\mathrm{sub}}$, so long as we choose $M_\varepsilon$ as stated above, even though Proposition~\ref{t:cc} cannot be applied directly.

\section{From boundary data to slope estimates --- proof of Theorem~\ref{t:main}}\label{s:5}
Using the sub- and supersolutions from \S\ref{s:3} and \S\ref{s:4} that bound the actual solution pointwise, we wish to obtain gradient estimates on the actual solution. For this,  we rely on a somewhat elementary lemma.  
First, note that by Proposition~\ref{t:cc} we have
\begin{align*}
    |\varphi_{\mathrm{sup}}(t,0)-\varphi(t,0)|\leq&|\varphi_{\mathrm{sup}}(t,0)-\varphi^\eps_{\mathrm{sub}}(t,0)|\\
    =&\big| \beta-\frac{2(1-\mu)\sqrt{t}}{\sqrt{\pi}}+2\pi\bigg\lfloor\frac{1-\mu+\eps}{\pi^{3/2}}\sqrt t\bigg\rfloor+\frac{\pi}{2}+2\pi M_\eps \big|\\
    \leq&\frac{2\eps}{\sqrt{\pi}}\sqrt{t}+|\beta|+\frac{\pi}{2}+2\pi M_\eps.
\end{align*}
Therefore by the above inequality and linearity, $\varphi_{\mathrm{sup}}-\varphi$ satisfies
\begin{equation}\label{e:bound}
    \begin{cases}
        \partial_t(\varphi_{\mathrm{sup}}-\varphi) = \partial_{xx} (\varphi_{\mathrm{sup}}-\varphi), & t,x>0,\\
        \varphi_{\mathrm{sup}}(t,0)-\varphi(t,0)=A(t), & t>0,\\
        \varphi_{\mathrm{sup}}(0,x)-\varphi(0,x)=0, & x>0,
    \end{cases}
\end{equation}
with $|A(t)|\lesssim\eps\sqrt{t}+M_\eps$. This motivates the following lemma.
\begin{lemma}\label{l:bound}
    Let $\psi$ be a solution to 
    \begin{equation}\label{e:lemma}
        \begin{cases}
        \partial_t \psi=\partial_{xx} \psi, & t,x>0,\\
        \psi(t,0)=a(t), & t>0,\\
        \psi(0,x)=0, & x>0.
    \end{cases}
    \end{equation}
    \begin{enumerate}
        \item If $|a(t)|\leq\sqrt{t},$ then there exists $C>0$ such that $|\partial_x\psi(t,\bar x\sqrt{t})|\leq \dfrac{C}{\bar{x}}$ for each $\bar x>0, t>0.$
        \item If $|a(t)|\leq1,$ then there exists $D>0$ such that $|\partial_x\psi(t,\bar x\sqrt{t})|\leq \dfrac{D}{\bar x\sqrt{t}}$ for each $\bar x>0, t>0.$
    \end{enumerate}
\end{lemma}
\begin{proof}
    First recall the boundary integral reconstruction~\eqref{e:recsoln}, which gives
    \begin{equation*}
        \psi(t,x)=\int_0^t\frac{x}{\sqrt{4\pi s^3}}e^{-x^2/4s}a(t-s) ds.
    \end{equation*}
    Taking derivatives and absolute values then yields
    \begin{align*}
    |\partial_x\psi(t,x)|\leq & \Bigg|\partial_x\int_0^t\frac{x}{\sqrt{4\pi s^3}}e^{-x^2/4s}a(t-s) ds\Bigg|\\
    \leq &\int_0^t\Bigg|\partial_x\bigg(\frac{x}{\sqrt{4\pi s^3}}e^{-x^2/4s}\bigg)a(t-s) \Bigg|ds\\
    =&\int_0^t\Bigg|\frac{1}{\sqrt{4\pi s^3}}-\frac{x^2}{4\sqrt{\pi s^5}}\Bigg|e^{-x^2/4s}|a(t-s)| ds.
    \end{align*}
    We now consider case (i) from the statement, first.
    Suppose $|a(t)|\leq\sqrt{t}$ for all $t>0.$ 
    Then applying the triangle inequality yields
    \begin{equation*}
        |\partial_x\psi(t,x)|\leq\int_0^t\Bigg(\frac{1}{\sqrt{4\pi s^3}}+\frac{x^2}{4\sqrt{\pi s^5}}\Bigg)e^{-x^2/4s}\sqrt{t-s} \ ds = \frac{2\sqrt{t}}{x}e^{-x^2/4t}+\frac{\sqrt{\pi}}{2}\Big(1-\erf\Big(\frac{x}{2\sqrt{t}}\Big)\Big).
    \end{equation*}
    Writing $x=\bar x\sqrt{t},$ we can find $C>0$ such that 
    \begin{equation*}
        |\partial_x\psi(t,\bar x\sqrt{t})|\leq \frac{2}{\bar x}e^{- \bar x^2/4}+\frac{\sqrt{\pi}}{2}\Big(1-\erf\Big(\frac{\bar x}{2}\Big)\Big)\leq\frac{C}{\bar x}.
    \end{equation*}
    We next turn to case (ii). Suppose $|a(t)|\leq1$ for all $t>0.$ We use a similar argument, arriving at the inequality
    \begin{equation*}
        |\partial_x\psi(t,x)|\leq\int_0^t\Bigg(\frac{1}{\sqrt{4\pi s^3}}+\frac{x^2}{4\sqrt{\pi s^5}}\Bigg)e^{-x^2/4s} \ ds=\frac{1}{\sqrt{\pi t}}e^{-x^2/4t}+\frac{2}{x}\Big(1-\erf\Big(\frac{x}{2\sqrt{t}}\Big)\Big).
    \end{equation*}
    Writing $x=\bar x\sqrt{t},$ we can find $D>0$ such that 
    \begin{equation*}
        |\partial_x\psi(t,\bar x\sqrt{t})|\leq\frac{1}{\sqrt{t}}\Bigg(\frac{1}{\sqrt{\pi}}e^{-\bar x^2 /4}+\frac{2}{\bar x}\Big(1-\erf\Big(\frac{\bar x}{2}\Big)\Big)\Bigg)\leq\frac{D}{\bar x\sqrt{t}}.
    \end{equation*}
    This concludes the proof.
\end{proof}
We can now turn to the proof of our main result, first in the case of zero initial slope.
\begin{proof}[Proof of Theorem~\ref{t:main}, $\alpha=0$]
    Let $\eps>0$. If we take $a(t) = \varphi_{\mathrm{sup}}(t,0)-\varphi(t,0)$ in~\eqref{e:lemma} then 
    $$\psi(t,x) = \varphi_{\mathrm{sup}}(t,x) - \varphi(t,x),$$
    is a solution to~\eqref{e:lemma}. Using the estimate on $A$ in~\eqref{e:bound}, Lemma~\ref{l:bound}, and~\eqref{e:supderiv}, we have that
    \begin{equation*}
        \Big|\partial_x\big[\varphi_{\mathrm{sup}}(t,\bar x\sqrt{t})-\varphi(t,\bar x\sqrt{t})\big]\Big|=\Big|\partial_x\varphi(t,\bar x\sqrt{t})-(1-\mu)\Big(1-\erf\bigg(\frac{\bar x}{2}\bigg)\Big)\Big|\lesssim  \bigg(\frac{\eps}{\bar{x}}+\frac{M_\eps}{\bar x \sqrt{t}}\bigg).
    \end{equation*}
    Now choose $\bar x=\sqrt \eps$ to find
    \begin{equation*}
        \Big|\partial_x\varphi(t,\sqrt{\eps t})-(1-\mu)\Big(1-\erf\bigg(\frac{\sqrt{\eps}}{2}\bigg)\Big)\Big|\lesssim \frac{\eps}{\sqrt{\eps}}+\frac{M_\eps}{\sqrt{\eps t}}=\sqrt{\eps}+\frac{M_\eps}{\sqrt{\eps t}}.
    \end{equation*}
    We choose $T>0$ sufficiently large such that for each $t>T$,
        $\frac{M_\eps}{\sqrt{t}}<\eps$ and hence $\frac{M_\eps}{\sqrt{\eps t}}<\sqrt{\eps}$.
    Then
    \begin{equation*}
        \Big|\partial_x\varphi(t,\bar x\sqrt{t})-(1-\mu)\Big(1-\erf\bigg(\frac{\sqrt{\eps}}{2}\bigg)\Big)\Big|\lesssim \sqrt{\varepsilon}.
    \end{equation*}
    With the reverse triangle inequality, we find 
    \begin{equation*}
        \Bigg|\Big|\partial_x\varphi(t,\bar x\sqrt{t})-(1-\mu)\Big| -\Big|(1-\mu)\Big(\erf\bigg(\frac{\sqrt{\eps}}{2}\bigg)\Big)\Big|\Bigg|\leq \Bigg|\partial_x\varphi(t,\bar x\sqrt{t})-(1-\mu)\Big(1-\erf\bigg(\frac{\sqrt{\eps}}{2}\bigg)\Big)\Big|\lesssim\sqrt{\eps}.
    \end{equation*}
    Since $\erf(z)>0$ when $z>0$, we can remove the absolute values on the second term. Then Taylor expanding the error function at zero, we arrive at 
    \begin{equation*}
        -\sqrt{\eps}\lesssim\Big|\partial_x\varphi(t,\sqrt{\eps t})-(1-\mu)\Big| -\Big((1-\mu)\Big(\frac{\sqrt{\eps}}{\sqrt{\pi}}+\mathcal O(\eps)\Big)\Big)\lesssim\sqrt{\eps}.
    \end{equation*}
    Therefore adding the second term across yields
    \begin{equation*}
        -\sqrt{\eps}+\mathcal O(\eps)\lesssim\Big|\partial_x\varphi(t,\sqrt{\eps t})-(1-\mu)\Big|\lesssim\sqrt{\eps}+\mathcal{O}(\eps),
    \end{equation*}
    and so
    \begin{equation*}
        \Big|\partial_x\varphi(t,\sqrt{\eps t})-(1-\mu)\Big|\lesssim\sqrt{\eps}.
    \end{equation*}
\end{proof}
It remains to adapt this proof to more general initial conditions.
\begin{proof}[Proof of Theorem~\ref{t:main}]
    The general case can be reduced to the specific case $\alpha=0$ by changing variables, scalings, and some elementary estimates. First suppose $\alpha < 1- \mu$. We define
    \[
    \varphi(t,x) = \bar \varphi(t,x) + \alpha x.
    \]
    Plugging this into~\eqref{PD}, we find that $\bar \varphi$ satisfies 
    \[
    \partial_t \bar \varphi = \partial_{xx} \bar \varphi, \quad \partial_x \bar \varphi(t,0) = 1-\alpha + \mu \sin(\bar \varphi(t,0)), \quad \text{and} \quad \bar \varphi(0,x)=\beta.
    \]
    Next, let
    $$x = \gamma \xi, \quad t = \gamma^2 \tau, \quad \text{and} \quad \tilde \varphi(\tau, \xi) = \bar \varphi(t,x).$$
    Then $\tilde \varphi$ satisfies
    \begin{equation}\label{e:tildesub}
        \partial_{\tau} \tilde \varphi = \partial_{\xi\xi} \tilde \varphi, \quad \partial_{\xi} \tilde \varphi(\tau,0) = \gamma(1-\alpha) + \gamma \mu \sin(\tilde \varphi(\tau,0)), \quad \text{and} \quad \tilde \varphi(0,\xi) = \beta.
    \end{equation}
    By choosing $\gamma = \dfrac{1}{1-\alpha}$, the boundary condition at $\xi=0$ reduces to
    $$\partial_{\xi} \tilde \varphi(\tau,0) = 1 + \frac{\mu}{1-\alpha} \sin(\tilde \varphi(\tau,0)).$$
    Since $\alpha < 1-\mu$, it follows that
    $$\tilde \mu = \frac{\mu}{1-\alpha} \in (0,1),$$
    and so this is simply the equation we studied above. So by Theorem~\ref{t:main}, with $\alpha=0$, it follows that for all $\varepsilon>0$, there exists a $T>0$ such that
    \[
        \Big\rvert\partial_{\xi}\tilde \varphi(\tau, \sqrt{\varepsilon \tau}\,) - (1-\tilde \mu)\Big|\lesssim \sqrt{\varepsilon}  \quad \forall \, t > T.
    \]
    Working backwards, $\xi = \sqrt{\varepsilon \tau}$ implies $x = \sqrt{\varepsilon t}$. So we readily conclude that
    \begin{align*}
        \Big\rvert \frac{1}{1-\alpha} \partial_{x}\bar \varphi(t, \sqrt{\varepsilon t}\,) - (1-\tilde\mu)\Big|\lesssim\sqrt{\varepsilon}  & \implies \Big\rvert \partial_{x}\bar \varphi(t, \sqrt{\varepsilon t}\,) - (1-\alpha)(1-\tilde\mu)\Big|\lesssim\sqrt{\varepsilon}\\
        & \implies \Big\rvert \partial_{x}\bar \varphi(t, \sqrt{\varepsilon t}\,) - (1-\alpha - \mu)\Big|\lesssim\sqrt{\varepsilon}.
    \end{align*}
    Finally undoing the last substitution we find that
    $$ \Big| \partial_{x}  \varphi(t, \sqrt{\varepsilon t}\,) - \alpha - (1-\alpha-\mu)\Big| = \Big| \partial_{x}  \varphi(t, \sqrt{\varepsilon t}\,) - (1-\mu)\Big| \lesssim\sqrt{\varepsilon} \quad \forall \, t>T,$$
    as desired.

    For the case $\alpha > 1+ \mu$, we start by setting
    $\varphi(t,x) = - \psi(t,x)$.
    Substituting this into~\eqref{PD} gives us
    $$\partial_t \psi = \partial_{xx} \psi, \quad \partial_x \psi(t,0) = -1- \mu \sin(-\psi(t,0)), \quad \text{and} \quad \bar \psi(0,x)=-\alpha x.$$
    Using that sine is odd, we reduce the boundary condition to 
    $$\partial_x \psi(t,0) = -1 + \mu \sin(\psi(t,0)).$$
    With the substitution
    $\psi(t,x) = \bar \psi(t,x) - \alpha x,$ and
    $x=\gamma \xi$, $ t = \gamma^2 \tau$, $\bar \psi(t,x) = \tilde \psi(\tau, \xi)$,
    we find that
    \begin{equation}\label{tilde2eqn}
        \partial_{\tau} \tilde \psi = \partial_{\xi\xi} \tilde \psi, \quad \partial_{\xi} \tilde \psi(\tau,0) = \gamma(\alpha-1) + \gamma \mu \sin(\tilde \psi(\tau,0)), \quad \text{and} \quad \tilde \psi(0,\xi) = \beta.
    \end{equation}
    Then since $\alpha > 1+\mu$, it follows that
    $$\tilde \mu = \frac{\mu}{\alpha -1} \in (0,1),$$
    just like above. Then again~\eqref{tilde2eqn} is exactly of the form we studied before, so Theorem~\ref{t:main} for $\alpha=0$ gives us that for all $\varepsilon>0$, there exists a $T>0$ such that
    $$\Big\rvert\partial_{\xi}\tilde \psi(\tau, \sqrt{\varepsilon \tau}\,) - (1-\tilde \mu)\Big|\lesssim \sqrt{\varepsilon} \quad \forall \, t > T.$$
    Then working backwards just as we did above yields that
    $$\Big| \partial_{x}  \varphi(t, \sqrt{\varepsilon t}\,) - (1+\mu)\Big| \lesssim\sqrt{\varepsilon}  \quad \forall \, t>T,$$
    as desired.
\end{proof}

\section{Discussion and outlook}\label{s:dis}

In summary, we showed that for a large class of initial  conditions to the phase diffusion approximation~\eqref{PD}, the derivative of solutions will be well approximated by  the marginally stable slope along steep parabolas $t = x^2/\eps$ inside space-time, visible as contour lines in Figures~\ref{f:sh} and~\ref{f:pd}. This echoes the numerical observation of convergence of average wavenumbers in large regions. We mention here some open questions.

\medskip\textbf{Alternative statements on convergence of slope.}
A different kind of convergence was illustrated in Figure~\ref{f:subsoln}, right panel, where averages of the limiting slope are plotted as length of average and time go to infinity.  Understanding convergence of these averages, one would start by examining the functions $\varphi_{\mathrm{sup}}$ and $\varphi^\varepsilon_{\mathrm{sub}}$ from~\eqref{e:supsoln} and~\eqref{e:subsoln} respectively. One finds that
\begin{equation*}
    \lim_{t \to \infty} \dfrac{\varphi_{\mathrm{sup}}(t,L)-\varphi_{\mathrm{sup}}(t,0)}{L} = 1-\mu,
\end{equation*}
and similarly for any $\varepsilon>0$
\begin{equation}\label{e:limsupinf}
    \liminf_{t \to \infty} \dfrac{\varphi^\varepsilon_{\mathrm{sub}}(t,L)-\varphi^\varepsilon_{\mathrm{sub}}(t,0)}{L} = 1-\mu+\varepsilon \quad \text{and} \quad \lim_{L \to \infty} \limsup_{t \to \infty} \dfrac{\varphi^\varepsilon_{\mathrm{sub}}(t,L)-\varphi^\varepsilon_{\mathrm{sub}}(t,0)}{L} = 1-\mu+\varepsilon.
\end{equation}
The latter can be established following the proof of Lemma~\ref{l:sum}, although taking the limit when applying squeeze theorem is significantly more involved. Then, using these  limits for sub- and supersolutions and similar ideas to those seen in the proofs of Lemma~\ref{l:bound} and Theorem~\ref{t:main}, one can show that for any $\varepsilon>0$ there exists a $T>0$ such that 
\begin{equation*}
    \bigg| \frac{\varphi(t,\sqrt{\eps t}) - \varphi(t,0)}{\sqrt{\eps t}} - (1-\mu) \bigg| \lesssim \sqrt{\varepsilon} \quad \text{for all } t>T.
\end{equation*}
However, showing that the uncoupled limits in~\eqref{e:limsupinf} hold for the actual solution with $\eps=0$, appears to be more difficult.


\medskip\textbf{Growth and marginal stability.} 
In many applied scenarios, the mismatch between wavenumber in the bulk and selected wavenumbers at the boundary results from growth, isotropic (that is, uniform stretching within the domain) or apical (the domain grows through adding at the boundary). The system presented here was analyzed in~\cite{PhysRevE.94.022219}, finding that slow growth selects wavenumber close to the extremal boundary, with universal corrections proportional to the square root of the speed of growth; see also~\cite{MR4646017}. Phenomenologically, the domain growth stretches patterns until the wavenumber is too small, becoming incompatible with the boundary condition, at which point snapping events adjust the wavenumber. The results there only provided formal asymptotics for the universal wavenumber corrections. It would be interesting to explore if subsolutions of the type used here could help establish asymptotics rigorously. 

\medskip\textbf{Bounded domains.} In a bounded domain, equilibria are isolated and generically hyperbolic. When considering the equation with values $\varphi\in \R/(2\pi\Z)$, the set of equilibria is compact, and so is the set of bounded solutions, sometimes referred to as the global attractor. The equilibria can be obtained somewhat explicitly, or equivalently through a somewhat intuitive geometric construction outlined in~\cite[\S2.6]{MorisseyScheel}, which also would give stability.  More refined information on the set of equilibria, particularly a permutation induced by the ordering on equilibria induced by a shooting construction, can give information on the structure of the global attractor; see~\cite{MR4710809} and references therein. It would be interesting to identify the role of the ``marginal'' equilibria in this finite-domain setting, particularly their location inside the attractor. 

\medskip\textbf{Fragility and statistics.} 
Self-organized criticality is of particular interest in connection with the statistics of systems when subject to noise. As we pointed out, the limiting state with critical slope is unstable and small perturbations will induce snapping. However, memory in the system, through the diffusive field, will make repeated snapping more difficult. We hope that the simple system analyzed here will be amenable to a more rigorous analysis of the influence of stochastic perturbations.

\medskip\textbf{Snapping beyond the phase-diffusion approximation.} 
We reemphasize that our methods are restricted to the phase diffusion approximation, although similar phenomena can be observed  in more realistic systems. We demonstrated the convergence in the Swift-Hohenberg equation, Figure~\ref{f:sh}, but suspect that other systems will exhibit similar phenomena; see~\cite{PhysRevE.94.022219} for applications to reaction-diffusion systems, the complex Ginzburg-Landau, and the Cahn-Hilliard equation. While these systems can be approximated by phase-diffusion dynamics in a slowly varying situation, snapping events intrinsically violate this assumption. On the other hand, comparison principles that allowed us to control snapping with rough $L^\infty$ bounds are not available in these contexts and more refined estimates on gluing snapping with slow drift along equilibria would likely be needed.

\medskip\textbf{Wavenumber selection.} 
The phenomena presented here can also be viewed in a broader context of pattern selection; see~\cite{MR5016800} for a recent review. The role of boundaries here can also be played by defects, which at times select energy minimizing wavenumbers that in turn also are marginally stable against a wrinkling instability~\cite{MR3603778}. We already mentioned domain growth~\cite{MR4646017}  and the related front invasion~\cite{avery2026selectionmechanismsinvasion} scenario as alternate selection mechanisms. Both can be understood as akin to sources of wave trains in non-equilibrium systems, such as the prominent spiral waves~\cite{MR4580296} or one-dimensional sources~\cite{MR2067899}. Both can be seen as marginally stable limiting cases of wave sinks, akin to Lax shocks, which would be observed with well-prepared initial conditions at spatial infinity. 

\bibliographystyle{abbrv} 
\bibliography{References}  

\end{document}